\documentclass[fleqn,a4paper]{article}
\def\NAME{main}

\usepackage[utf8]{inputenc}
\usepackage[british]{babel}
\usepackage[T1]{fontenc}
\usepackage{cmbright} 
\usepackage{fullpage}

\usepackage{url}
\usepackage{html}

\usepackage{prettyref}
\newrefformat{def}{Def.~\ref{#1}}
\newrefformat{fig}{Fig.~\ref{#1}}
\newrefformat{rem}{Remark~\ref{#1}}
\newrefformat{prop}{Prop.~\ref{#1}}
\newrefformat{lem}{Lemma~\ref{#1}}
\newrefformat{cor}{Corollary~\ref{#1}}
\newrefformat{sec}{Sect.~\ref{#1}}
\newrefformat{ssec}{Subsect.~\ref{#1}}
\newrefformat{suppl}{Appendix~\ref{#1}}
\newrefformat{tab}{Table~\ref{#1}}
\newrefformat{suppl}{Appendix~\ref{#1}}
\newrefformat{eq}{Eq.~\eqref{#1}}
\def\pref{\prettyref}

\usepackage{graphicx}

\usepackage{amssymb}

\usepackage{array}
\usepackage{amsmath}

\usepackage{amsthm}
\newtheorem{remark}{Remark}
\newtheorem{theorem}{Theorem}[section]
\newtheorem{lemma}[theorem]{Lemma}
\newtheorem{corollary}[theorem]{Corollary}
\newtheorem{definition}[theorem]{Definition}
\newtheorem*{theorem*}{Theorem}

\theoremstyle{example}
\newtheorem*{example}{Example}

\def\DEF{\stackrel{\Delta}=}
\def\EQDEF{\stackrel{\Delta}\Leftrightarrow}

\def\thetitle{Dynamical Properties of Discrete Reaction Networks}
\title{\thetitle}
\author{Lo\"ic Paulev\'e$^1$, Gheorghe Craciun$^2$, Heinz Koeppl$^1$}
\date{$^1$ ETH Zürich\\
$^2$ University of Wisconsin-Madison}

\usepackage{tikz}
\pgfrealjobname{\NAME}
\usetikzlibrary{arrows}
\tikzstyle{axe}=[->,>=stealth]
\tikzstyle{vector}=[ultra thick,->,>=latex]

\def\V{\mathcal V}
\def\O{\mathcal O}

\def\zero{\mathbf 0}
\def\R{\mathbb R}
\def\Rpos{\R_{\geq 0}}
\def\Rspos{\R_{>0}}
\def\Q{\mathbb Q}

\def\Qspos{\Q_{>0}}
\def\Z{\mathbb Z}
\def\Zpos{\Z_{\geq 0}}
\def\Zspos{\Z_{> 0}}

\def\mspan{\mathrm{span}}
\def\lp{\mathrm{lowerpoint}}
\def\orderings{\mathrm{orderings}}

\def\inner#1#2{#1#2}

\def\apply{\bullet}

\def\B{\mathcal B}

\usepackage{natbib}

\begin{document}

\maketitle

\begin{abstract}
Reaction networks are commonly used to model the evolution of populations of species subject to
transformations following an imposed stoichiometry. 

This paper focuses on the efficient characterisation of dynamical properties of \emph{Discrete
Reaction Networks} (DRNs).
DRNs can be seen as modelling the underlying discrete nondeterministic transitions of stochastic
models of reactions networks.
In that sense, any proof of non-reachability in DRNs directly applies to any concrete
stochastic models, independently of kinetics laws and constants.
Moreover, if stochastic kinetic rates never vanish, reachability properties are equivalent in the
two settings.

The analysis of two global dynamical properties of DRNs is addressed:
irreducibility, i.e., the ability to reach any discrete state from any other state;
and recurrence, i.e., the ability to return to any initial state.
Our results consider both the verification of such properties when species are present in a large
copy number, and in the general case.
The obtained necessary and sufficient conditions involve algebraic conditions on the network
reactions which in most cases can be verified using linear programming.

Finally, the relationship of DRN irreducibility and recurrence with dynamical properties of
stochastic and continuous models of reaction networks is discussed.
\end{abstract}


\section{Introduction} 

Reaction networks describe the possible transformations between species in a system, subject to
stoichiometry constraints (e.g. $2A + B \rightarrow C +D$).
They are widely used for fine-grained modelling of various complex dynamical system, and in particular biochemical
systems.
Typically, reaction network models are equipped with kinetic laws in order to take into account the
influence of the various speeds and propensities of the involved reactions on the overall dynamics.
Depending on the nature of the systems and interacting species, those kinetics may follow different
laws.
These reaction networks and kinetic rules are then generally interpreted either in continuous
frameworks, such as ODEs \citep{FeinbergLecture,CTF-PNAS06}, which relates the dynamics of the concentration of the species;
or in stochastic frameworks, such as continuous-time Markov chains \citep{Wilkinson06,ACK-BMB10}, which precisely track the
population (copy number) of each species along the time.

In practice, such modelling techniques face two challenges:
the actual kinetics are most often unknown and may substantially vary between systems sharing the
same reaction network;
and formal analysis of the emerging dynamical properties is computationally intractable for large-scale continuous
and stochastic models.

In this paper, we propose a more abstract level of interpretation of reaction networks,
by focusing on the nondeterministic discrete evolution of the population of the species.
Given the population of each species (discrete state), the system can evolve along the application
of any reaction where the minimum amount of copy number of transformed species is present.
We consider that only one discrete reaction can be applied at a time.
Such nondeterministic systems can be formally considered as the discrete underlying dynamics of
stochastic models of reaction networks \citep{FS08}.

In such a setting, dynamics of \emph{Discrete Reaction Networks} (DRNs) naturally delimit the dynamics of
concrete stochastic systems, whatever the kinetic laws and constants:
if a reachability is proved impossible in a DRN, it is also impossible for any particular stochastic
model of the network.
In the case where the rate (or probability) of a reaction in the stochastic model never
becomes zero, the (discrete) reachability properties of the stochastic model are equivalent with the
corresponding properties of the underlying DRN.
In general, one can think of a DRN as underlying any discrete stochastic model of the reaction
network.

Here, we demonstrate that some general dynamical reachability properties
can be efficiently derived from a DRN:
the capacity to reach any discrete state from any other state (irreducibility);
and the reversibility of the reachability properties (recurrence).
Such properties are both considered in the case where species are present in a large copy number as
well as in the general case.
These results help provide an understanding of the possible global dynamics of reaction networks,
and give a direct relationship between the structure of the set of reactions and the verification of
the mentioned dynamical properties, without any assumption on kinetic laws.

\bigskip

The main objects and results presented in this paper are summarised below.

\paragraph{Notations.}
For any $a,b$ in $\Z$, $[a;b]$ denotes the set of integers between $a$ and $b$ that is
	$\{ a, a+1, \dots, b \}$.
For any $x,x'$ in $\Z^d$, $x$ is \emph{greater} than $x'$, denoted $x \succeq x'$, if and only if
every component of $x$ is greater than the corresponding component in $x'$, i.e., for any $i$ in
$[1;d]$, $x_i \geq x'_i$.
The set of matrices of elements in $G$ having $n$ lines and $d$ columns is denoted by
$G^{n\times d}$.
If $\V$ is in $G^{n\times d}$, for any $j\in[1;n]$, $V_j$ is the $j^{\text{th}}$ line, and
$\V_j$ is in $G^d$.
Given a field $F$, and a matrix $\V \in G^{n\times d}$,
the span of $\V$ over $F$ is denoted by
$\mspan_F \V \DEF \{ \inner{\lambda}{\V} \mid \lambda\in F^n \}$.
Finally, the null vector is referred to as $\zero$.

\paragraph{Discrete Reaction Networks}

We consider a set of reactions between $d$ species $A_i, i\in[1;d]$ of the form
\begin{equation}
c_1 A_1 + \dots + c_d A_d \longrightarrow c'_1 A_1 + \dots + c'_d A_d
\label{eq:reaction}
\end{equation}
where for any $i$ in $[1;d]$,  $c_i$ and $c'_i$ are in $\Zpos$.
The reaction can be applied as soon as the population of species $A_i$ is at least $c_i$, for any
$i$ in $[1;d]$.
Its application decreases the population of species $A_i$ by $c_i$ and then increases it by $c'_i$.
Such a reaction can be summarised by two vectors of dimension $d$:
$v = (c'_1-c_1, \cdots, c'_d-c_d)$, the \emph{drift} vector describing the population changes after
application of the reaction;
and $o = (c_1, \cdots, c_d)$, the \emph{origin} of the reaction, that is the minimum required
population of species for applying the reaction.

In such a setting, a \emph{Discrete Reaction Network} (DRN) of $n$ reactions between $d$ species 
can be defined by a couple $(\V,\O)$ of two matrices having $d$ columns and $n$ rows:
$\V$ gathers the drift vectors of the $n$ reactions and $\O$ their origins (\pref{def:DRN}).
We impose that each reaction can be applied at least once from its origin, i.e. the population of
species does not reach negative values.

\begin{definition}[Discrete Reaction Network]
\label{def:DRN}
A \emph{Discrete Reaction Network} (DRN) is a couple $(\V,\O)$, where
$\V\in \Z^{n\times d}$,
$\O\in \Zpos^{n\times d}$,
and $\forall i\in [1;n]$, $\O_i+\V_i\succeq\zero$.
$n$ is the \emph{size} and $d$ is the \emph{dimension} of the DRN.
\end{definition}

\begin{example}
\pref{fig:examples} shows two examples of DRNs with 3 reactions between 2 species.
\begin{itemize}
\item 
Example (a).
reactions:
$
\begin{array}{rcl}
\varnothing & \rightarrow & 2A\\
A + B & \rightarrow & \varnothing\\
5A & \rightarrow &4A + 3B
\end{array}
\Rightarrow
\V = \left(
\begin{array}{rr}
2 & 0\\
-1 & -1\\
-1 & 3
\end{array}
\right),
\O = \left(
\begin{array}{rr}
0 & 0\\
1 & 1\\
5 & 0
\end{array}
\right)$
\enspace.
\item 
Example (b).
reactions:
$
\begin{array}{rcl}
\varnothing & \rightarrow & 2A\\
A + B & \rightarrow & \varnothing\\
5A & \rightarrow &4A + 2B
\end{array}
\Rightarrow
\V = \left(
\begin{array}{rr}
2 & 0\\
-1 & -1\\
-1 & 2
\end{array}
\right),
\O = \left(
\begin{array}{rr}
0 & 0\\
1 & 1\\
5 & 0
\end{array}
\right)$
\enspace.
\end{itemize}
\end{example}

\begin{figure}
\centering

\begin{tabular}{c@{\hskip 1in}c}

\begin{tikzpicture}

\draw[help lines] (-0.2,-0.2) grid (5.8,4.8);

\draw[axe] (0,0) -- (5.8,0);
\draw[axe] (0,0) -- (0,4.8);
\node[anchor=north east] at (5.8,0) {$A$};
\node[anchor=north east] at (0,4.8) {$B$};

\node[anchor=north east] at (0,0) {$\zero$};

\draw[vector] (0,0) -> (2,0) node [sloped, above, near end] {$\V_1$};
\fill (0,0) circle (2pt);
\node[anchor=north west] at (0,0) {$\O_1$};

\draw[vector] (1,1) -> (0,0) node [sloped,above,midway] {$\V_2$};
\fill (1,1) circle (2pt);
\node[anchor=south west] at (1,1) {$\O_2$};

\draw[vector] (5,0) -> (4,3) node [sloped,above,midway] {$\V_3$};
\fill (5,0) circle (2pt);
\node[anchor=south west] at (5,0) {$\O_3$};

\end{tikzpicture}

&

\begin{tikzpicture}

\draw[help lines] (-0.2,-0.2) grid (5.8,4.8);

\draw[axe] (0,0) -- (5.8,0);
\draw[axe] (0,0) -- (0,4.8);
\node[anchor=north east] at (5.8,0) {$A$};
\node[anchor=north east] at (0,4.8) {$B$};
\node[anchor=north east] at (0,0) {$\zero$};

\draw[vector] (0,0) -> (2,0) node [sloped, above, near end] {$\V_1$};
\fill (0,0) circle (2pt);
\node[anchor=north west] at (0,0) {$\O_1$};

\draw[vector] (1,1) -> (0,0) node [sloped,above,midway] {$\V_2$};
\fill (1,1) circle (2pt);
\node[anchor=south west] at (1,1) {$\O_2$};

\draw[vector] (5,0) -> (4,2) node [sloped,above,midway] {$\V_3$};
\fill (5,0) circle (2pt);
\node[anchor=south west] at (5,0) {$\O_3$};

\end{tikzpicture}
\\
Example (a) & Example (b)
\end{tabular}

\caption{Two DRNs with 3 reactions between 2 species $A$ and $B$.}
\label{fig:examples}
\end{figure}

We will see in \pref{sec:irreducibility} and \ref{sec:recurrence} that these similar-looking DRNs
have difference dynamical properties.

\paragraph{Discrete transitions}
The population of the $d$ species of the DRN forms a \emph{discrete state} (or point) of the DRN, and is represented as a
vector $x$ in $\Zpos^d$.
At $x$, only the reactions $j$ in $[1;n]$ such that $x \succeq \O_j$ can occur.
The occurrence of one of these reactions leads to the state $x'=x+\V_j$, with necessarily
$x'$ in $\Zpos$.
The transition relation $\rightarrow$ (\pref{def:DRN-Sem}) is defined such that
$x \rightarrow x'$ if and only if $x'$ can be reached by the occurrence of one (and only one) reaction from $x$.
The binary relation $\rightarrow^*$ extends the binary relation $\rightarrow$ by considering the
successive occurrence of any number of reactions.
Hence for any $x,x'$ in $\Zpos^d$, $x\rightarrow^* x'$ if and only if there exists a
sequence of reaction occurrences from $x$ leading to exactly $x'$ which never makes negative the population
of any species.

\begin{definition}[Transition relation $\rightarrow$]
\label{def:DRN-Sem}
Given a DRN $(\V,\O)$ and two points $x,x'\in\Zpos^d$, 
$x \rightarrow_{(\V,\O)} x'$ if and only if $\exists i\in[1;n]$ such that
$x \succeq \O_i$ and $x+\V_i = x'$.
$\rightarrow^*_{(\V,\O)}$ is the transitive closure of binary relation $\rightarrow_{(\V,\O)}$.
When clear from context, $\rightarrow_{(\V,\O)}$ is written as $\rightarrow$.
\end{definition}

\paragraph{Irreducibility and Recurrence}
In this paper, we focus on two dynamical properties of DRNs:
\begin{itemize}
\item \emph{Irreducibility}: a DRN is irreducible if and only if one can reach any point $x'\in\Zpos$ from any
point $x\in\Zpos$ (\pref{def:irreducibility}).

\item \emph{Recurrence}: a DRN is recurrent if and only if one can always reverse the application of
any sequence of reactions (\pref{def:recurrence}).

\end{itemize}
It is worth noticing that any irreducible DRN is recurrent (\pref{rem:irred-lt-recur}).

\begin{definition}[Irreducibility]
\label{def:irreducibility}
DRN $(\V,\O)$ is \emph{irreducible} if and only if
$\forall x,x'\in \Zpos^d$, $x\rightarrow^* x'$ and $x'\rightarrow^* x$.
\end{definition}

\begin{definition}[Recurrence]
\label{def:recurrence}
DRN $(\V,\O)$ is \emph{recurrent} if and only if
$\forall x,x'\in \Zpos^d$, $x\rightarrow^* x' \Longrightarrow x'\rightarrow^* x$.
\end{definition}

\begin{remark}
\label{rem:irred-lt-recur}
Irreducibility $\Longrightarrow$ Recurrence.
\end{remark}

In addition of considering irreducibility and recurrence from any possible population of species of
the DRN, we also investigate a weaker version of those dynamical properties when assuming the
species are present at a Large Copy Number (LCN).
This basically restricts the above dynamical properties to population of species greater than a
certain threshold $M_0$ in $\Zpos^d$.
We refer to these weaker properties as
\emph{LCN irreducibility} (\pref{def:lcn-irreducibility})
and \emph{LCN recurrence} (\pref{def:lcn-recurrence}), respectively.
Note that the inclusion relationship between irreducibility and recurrence still holds
(\pref{rem:lcn-irred-lt-recur}).

\begin{definition}[LCN Irreducibility]
\label{def:lcn-irreducibility}
DRN $(\V,\O)$ is \emph{LCN irreducible} if and only if
$\exists M_0\in\Zpos^d$ such that
$\forall x,x'\in \Zpos^d$ with $x\succeq M_0$ and $x'\succeq M_0$,
$x\rightarrow^* x'$ and $x'\rightarrow^* x$.
\end{definition}

\begin{definition}[LCN Recurrence]
\label{def:lcn-recurrence}
DRN $(\V,\O)$ is \emph{LCN recurrent} if and only if
$\exists M_0\in\Zpos^d$ such that
$\forall x,x'\in \Zpos^d$ with $x\succeq M_0$ and $x'\succeq M_0$,
$x\rightarrow^* x' \Longrightarrow x'\rightarrow^* x$.
\end{definition}

\begin{remark}
\label{rem:lcn-irred-lt-recur}
LCN Irreducibility $\Longrightarrow$ LCN Recurrence.
\end{remark}

\subsection*{Main Results}

In \pref{sec:irreducibility}, we first demonstrate that LCN irreducibility is equivalent to have both
the strictly positive real span of drift vectors being $\R^d$ and
the integer span of drift vectors being $\Z^d$.
\begin{theorem*}[\ref{thm:lcn-irreducibility}]
DRN $(\V,\O)$ is LCN irreducible if and only if $\mspan_{\Rspos} \V = \R^d$ and $\mspan_{\Z} \V =
\Z^d$.
\end{theorem*}
\noindent
Verifying $\mspan_{\Rspos} \V=\R^d$ can be done using linear programming, and verifying $\mspan_{\Z}
\V = \Z^d$ can be also efficiently done by computing, for instance, the Hermite normal form of $\V$.

Then, we show additional properties that lead to full irreducibility: self-starting (capability to
reach a strictly positive point from $\zero$) and self-stopping (capability to reach $\zero$ from a
strictly positive point).
\begin{theorem*}[\ref{thm:irreducibility}]
DRN $(\V,\O)$ is irreducible if and only if $(\V,\O)$ is LCN irreducible, self-starting and
self-stopping.
\end{theorem*}
\noindent
Self-starting and self-stopping properties can be decided using a backtracking algorithm combined
with linear programming to find a particular order of reactions

In \pref{sec:recurrence}, we prove that LCN recurrence is equivalent to the presence of $\zero$ in the strictly
positive real span of drift vectors.
Surprisingly, no integer constraints need to be checked, so this property can be easily decided
using linear programming.
\begin{theorem*}[\ref{thm:lcn-recurrent}]
DRN $(\V,\O)$ is LCN recurrent if and only if $\zero\in\mspan_{\Rspos} \V$.
\end{theorem*}

\bigskip

\pref{sec:examples} applies those results to DRNs modelling biological systems.
Presented results and their relationships with stochastic and continuous models of reaction networks are
discussed in \pref{sec:discussion}.
For example, we show how we can use the theorems above to check that common phosphorylation chain
networks are LCN recurrent and some circadian clock networks are LCN irreducible.

\section{Additional definitions, basic properties}
\label{sec:defs}

\subsection{Set of points and paths manipulation}

We introduce the following notations to manipulate set of points and paths (sequences of reactions):
\begin{description}
\item[lowerpoint]
Given a set of $m$ points $\{x_1, \dots, x_m\} \subset \Z^d$,
we denote by $\lp(\{x_1, \dots, x_m\})$ a point that is lower than all the given points:
\begin{align*}
\lp(\{x_1, \dots, x_m\}) & \DEF y \in \Z^d:
	\forall i\in[1;d], y_i = \min \{ x_{j,i}\mid j\in [1;m] \}\\
\end{align*}

\item[orderings] 
Given $\lambda\in\Zpos^n$ with $\ell=\sum_{i=1}^n \lambda_i$,
we denote by $\orderings(\lambda)$ all the mappings $\pi:[1;\ell]\mapsto [1;n]$ which map exactly
$\lambda_i$ distinct values to $i$, $\forall i\in[1;n]$:
\[
\orderings(\lambda) \DEF \{ \pi:[1;\ell]\mapsto[1;n] \mid
    \forall i\in[1;n], \lambda_i=\#\{ j \in[1;\ell]\mid \pi(j)=i\} \}
\]
where $\#\{e_1, \dots, e_k\} \DEF k$.

Hereafter, we use such mappings $\pi:[1;\ell]\mapsto [1;n]$ to refer to \emph{paths}, i.e. sequences
of reactions.
In such a context, $\lambda\in\Zpos^n$ should be understood as the vector giving the number of
times each reactions in $[1;n]$ has to be used in a path;
and $\orderings(\lambda)$ as all the possible permutations of such paths.

\item[path application ($x \apply \pi$)]
Given a DRN $(\V,\O)$ of size $n$ and dimension $d$, a path $\pi:[1;\ell]\mapsto [1;n]$, and an initial point $x\in\Z^d$,
$x\apply\pi$ is the set of points resulting from the sequential application of $\pi$ from $x$:
\[x\apply \pi \DEF \{ x + \sum_{i=1}^k v_{\pi(i)} \mid k\in [0;\ell] \}\enspace.\]
\end{description}

\subsection{Inverse DRN}

The \emph{inverse DRN} (\pref{def:1-DRN}) is defined by the negative drift vectors and the
origins shifted by the original drift vector.
For instance, the inverse of the reaction described in \pref{eq:reaction} results in:
\begin{equation}
(c_1+c'_1) A_1 + \dots + (c_d+c'_d) A_d \longrightarrow c_1 A_1 + \dots + c_d A_d
\end{equation}

\begin{definition}[Inverse DRN]
\label{def:1-DRN}
Given a DRN $(\V,\O)$, $(\V,\O)^{-1}\DEF(-\V,\O+\V)$ is the \emph{inverse DRN}.
\end{definition}
\begin{lemma}
$x \rightarrow_{(\V,\O)} x' \Longleftrightarrow x' \rightarrow_{(\V,\O)^{-1}} x$.
\label{lem:inverse}
\end{lemma}

\subsection{Basic properties}

From the definition of transitions between the discrete states of the DRN (\pref{def:DRN-Sem}), one can easily
derives that if $x \rightarrow^* x'$ then any succession of reactions from $x$ to $x'$ can be applied from
$x$ (positively) shifted by any $\delta\in\Zpos^d$, leading to $x'+\delta$ (\pref{lem:expansion}).
In the particular case when $\zero \rightarrow^* x'$, one can instantiate the latter property with
$\delta=x'$, which by transitivity of $\rightarrow$ leads to $\zero\rightarrow^* \alpha x'$ with $\alpha\in\Zspos$
(\pref{lem:scale}).

\begin{lemma}
\label{lem:expansion}
Given $x,x'\in\Zpos$,
$x\rightarrow^* x' \Longrightarrow \forall \delta \in \Zpos, x+\delta \rightarrow^* x'+\delta$.
\end{lemma}

\begin{lemma}
$\zero\rightarrow^* x' \Rightarrow \forall \alpha\in\Zspos, \zero\rightarrow^* \alpha x'$.
\label{lem:scale}
\end{lemma}

Given $x,x'\in\Zpos^d$ and $\lambda\in\Zpos^n$ such that $x' = x+\lambda\V$, it is sufficient (but
not necessary) to show that there exists $\pi\in\orderings(\lambda)$ verifying
$\forall j\in[1;n], \lp(x\apply\pi) \succeq \O_j$ to conclude that $x\rightarrow^* x'$.
We remark finally that $\lp(x\apply\pi) = x + \lp(\zero\apply\pi)$.

\section{Deciding Irreducibility}
\label{sec:irreducibility}

DRN $(\V,\O)$ is irreducible if any point in $\Zpos$ can be reached from any other point in $\Zpos$
(\pref{def:irreducibility}).
We first address the LCN irreducibility, and then exhibit supplementary properties that lead to 
full irreducibility.

\subsection{LCN Irreducibility}

Recall that DRN $(\V,\O)$ is LCN irreducible if and only if any point above a certain $M_0\in\Zpos^d$ can be
reached from any other point above $M_0$ (\pref{def:lcn-irreducibility}).

Before using the LCN hypothesis, we remark that the DRN is irreducible if (and only if) one can
reach each elementary point $e_i, \forall i\in[1;d]$ ($e_i$ is the $d$-dimensional vector having $0$ at
each of its component, except the $i^{\text{th}}$ component being $1$) from $\zero$ and vice-versa
(\pref{lem:irreducibility-ei}).

\begin{lemma}
\label{lem:irreducibility-ei}
DRN $(\V,\O)$ is irreducible if and only if $\forall i\in[1;d]$,
$\zero \rightarrow^* e_i$ and $e_i\rightarrow^*\zero$.
\end{lemma}

Note that a necessary condition for LCN irreducibility is that $\mspan_{\Zpos} \V = \Z^d$.
This property is actually sufficient for LCN irreducibility (\pref{lem:LCNir-posZspan}) by
choosing $M_0$ big enough such that for any $i\in[1;d]$ at least one reachability path from $M_0$ to $M_0 \pm e_i$
never goes below $\zero$, and such that $M_0$ is greater than all the reaction origins.

Remarking that $\mspan_{\Qspos} \V = \Q^d \Leftrightarrow \mspan_{\Rspos}\V = \R^d$
(\pref{lem:RQ}), \pref{thm:lcn-irreducibility} establishes that verifying
$\mspan_{\Zpos} \V =\Z^d$ is
equivalent to verifying
both $\mspan_{\Z} \V = \Z^d$ and $\mspan_{\Rspos}\V = \R^d$.

While the verification of $\mspan_{\Zpos}\V = \Z^d$ involves integer programming techniques,
verifying if $\mspan_{\Rspos} \V = \R^d$ and $\mspan_{\Z} \V = \Z^d$ can be done more
efficiently:
the former can be decided using linear programming, for instance by first checking if
$\zero\in\mspan_{\Rspos}\V$ and then if $\mspan_{\Rpos}\V = \R^d$;
the latter can be decided, for instance, by computing the Hermite normal form of $\V$
\citep{Cohen93}.

\begin{lemma}\label{lem:LCNir-posZspan}
DRN $(\V,\O)$ is LCN irreducible
$\Longleftrightarrow \mspan_{\Zpos} \V = \Z^d$.
\end{lemma}
\begin{proof}
$\mspan_{\Zpos} \V = \Z^d \Rightarrow \forall i\in[1;d],
	\exists \lambda^{i,+},\lambda^{i,-} \in \Zpos^n:
		\inner{\lambda^{i,+}}{\V} = e_i
		\wedge \inner{\lambda^{i,-}}{\V} = -e_i$.

For each $i\in[1;d]$ and $s\in\{+,-\}$,
we pick an arbitrary ordering $\pi\in\orderings(\lambda^{i,s})$.

If $M_0$ is defined such that
$\forall i\in[1;d],\forall s\in\{+,-\},\forall j\in[1;n]$,
$M_0 + \lp(\zero\apply\pi^{i,s}) \succeq \O_j$,
then 
it is clear that $\forall i\in[1;d]$,
	$M_0\rightarrow^* M_0+e_i$ and $M_0+e_i \rightarrow^* M_0$.
\end{proof}

\begin{lemma}
$\mspan_{\Rspos} \V = \R^d \Leftrightarrow \mspan_{\Qspos} \V = \Q^d$.
\label{lem:RQ}
\end{lemma}
\begin{proof}
Let us consider $\lambda\in\Rspos^n$ such that $\inner{\lambda}{\V} = w$, where $w\in\Q^d$.

Considering a basis $(\beta_\alpha)_{\alpha\in I}$ of $\R$ over $\Q$ such that $\beta_{\alpha_0} =
1$
(i.e. $\forall r\in\R, \exists\text{ a unique choice of }r^\alpha\in\Q: 
r=\sum_{\alpha\in I} r^\alpha \beta_\alpha$).
Then $\inner{\lambda}{\V}= \sum_{j=1}^n \lambda_j \V_j
= \sum_{j=1}^n (\sum_{\alpha\in I} \lambda^\alpha_j \beta_\alpha) \V_j
= \sum_{\alpha\in I} (\sum_{j=1}^n \lambda^\alpha_j \V_j)\beta_\alpha = w$
with $\lambda^\alpha\in \Q$.
On the other hand, $w = w\beta_{\alpha_0} + \sum_{\alpha\in I\setminus \{\alpha_0\}} 0\beta_\alpha$.
Hence, $\sum_{j=1}^n \lambda^{\alpha_0}_j \V_j = w$ and
$\forall \alpha\in I, \alpha\neq\alpha_0, \sum_{j=1}^n \lambda^\alpha_j \V_j = \zero$.

Therefore, $w = \sum_{j=1}^n \lambda_j^{\alpha_0} \V_j +
		\sum_{\alpha\in I \setminus \{\alpha_0\}} (\sum_{j=1}^n \lambda_j^{\alpha} \V_j )
		\tilde\beta_\alpha = 
		\sum_{\alpha\in I} (\sum_{j=1}^n \lambda_j^{\alpha} \V_j) \tilde\beta_\alpha =
		\sum_{j=1}^n (\sum_{\alpha\in I} \lambda_j^{\alpha}\tilde\beta_\alpha) \V_j$
with $\tilde\beta_{\alpha_0} = 1$ and $\forall \alpha\in I, \alpha\neq \alpha_0,
\tilde\beta_{\alpha} \in \Q$.
\end{proof}

\begin{theorem}
\label{thm:lcn-irreducibility}
$\mspan_{\Zpos} \V = \Z^d
\Longleftrightarrow
\mspan_{\R_{>0}} \V = \R^d \text{ and }
\mspan_{\Z} \V = \Z^d$
\end{theorem}
\begin{proof}
($\Leftarrow$)
$\mspan_{\R_{>0}} \V = \R^d \Leftrightarrow \mspan_{\Q_{>0}} \V = \Q^d$ (\pref{lem:RQ}).
Therefore, $\exists\lambda\in\Qspos^n$ such that 
	$\inner{\lambda}{\V} = \zero$ and $\exists \alpha\in\Z_{>0}$ such that
$\alpha\lambda\in\Z_{>0}^d$.
Moreover, $\forall i\in [1;d]$ and $\forall s\in\{+,-\}$,
	$\exists\lambda^{i,s} \in \Z^n$ such that
	$\inner{\lambda^{i,s}}{\V} = s e_i$.
Hence, there exists $\beta\in\Z_{>0}$ such that
	$\lambda^* = \beta\alpha\lambda + \lambda^{i,s}$
	with $\lambda^*\in \Zpos^d$,
	resulting in $\inner{\lambda^*}{\V} = s e_i$.
($\Rightarrow$)
straightforward by remarking that
$\mspan_{\Zpos} \V = \Z^d \Rightarrow \mspan_{\Zspos} \V = \Z^d$.
\end{proof}

\begin{example}
One can check that both examples of \pref{fig:examples} verify
$\mspan_{\Rspos}\V = \R^d$.
However, the computation of Hermite normal forms shows that only example (b) verifies the second necessary condition 
$\mspan_{\Z}\V = \Z^d$.
Hence, example (a) is not LCN irreducible whereas example (b) is LCN irreducible.
\end{example}

\subsection{Full Irreducibility}

In this subsection, we demonstrate that the DRN is totally irreducible if and only if the DRN is LCN
irreducible and is both \emph{self-starting} (\pref{def:self-starting}) and \emph{self-stopping}
(\pref{def:self-stopping}).
A DRN is self-starting if at least one strictly positive point can be reached from $\zero$, and is
self-stopping if there exists at least on strictly positive point from which $\zero$ can be reached
-- which is equivalent to the inverse DRN being self-starting.

\begin{definition}[Self-starting DRN]
\label{def:self-starting}
DRN $(\V,\O)$ is \emph{self-starting} if and only if
$\exists x \in \Zspos^d \text{ such that } \zero\rightarrow^* x$.
\end{definition}

\begin{definition}[Self-stopping DRN]
\label{def:self-stopping}
DRN $(\V,\O)$ is \emph{self-stopping} if and only if
inverse DRN $(\V,\O)^{-1}$ is self-starting.
\end{definition}

\pref{lem:self-starting} establishes that a DRN is self-starting if and only if there exists a
sequence of $d$ reactions (not necessarily unique) such that for each dimension at least one
reaction of this sequence has a positive drift along that dimension,
and such that the origin of the $k^{\text{th}}$ reaction belongs to the positive real span of the
$k-1$ preceding drift vectors (the first reaction having necessarily $\zero$ as origin).
Therefore, one can derive a backtrack algorithm to determine if such an ordering of reactions
exists.

Then, \pref{thm:irreducibility} states that if a LCN irreducible DRN is both self-starting and
self-stopping then it is irreducible.
Indeed, if the DRN is self-starting, then there exists a strictly positive point $x\in\Zspos$
such that $\zero\rightarrow^* x$. 
From \pref{lem:expansion}, this implies that $\zero\rightarrow^* x+M_0$.
Finally if the DRN is also self-stopping, one can easily show that there exists a point $x'\succeq
M_0$ such that $x' \rightarrow^* \zero$.
Because the DRN is LCN recurrent, we know that any pair of points above $M_0$ is reversibly
reachable.
Hence, by using \pref{lem:expansion}, one can verify the existence of a reversible path from
$\zero$ to all $e_i, i\in[1;d]$.

Informally, the self-starting property allows to reach the LCN region, and the self-stopping
allows to reach any $\pm e_i$ or $\zero$ from any point in the LCN region.
The LCN irreducibility property finally ensures that those two paths can be connected.
This is illustrated in \pref{fig:irreducibility}.

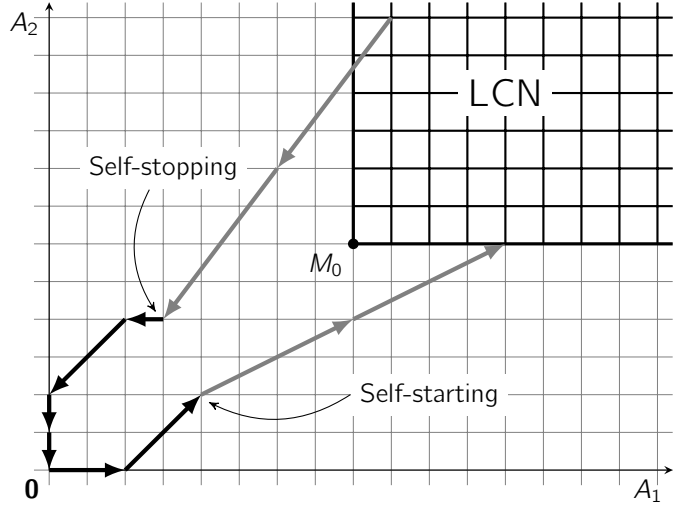
\begin{figure}
\centering

\begin{tikzpicture}

\draw[help lines] (-0.2,-0.2) grid[step=5mm] (8.2,6.2);

\draw[axe] (0,0) -- (8.2,0);
\draw[axe] (0,0) -- (0,6.2);
\node[anchor=north east] at (8.2,0) {$A_1$};
\node[anchor=north east] at (0,6.2) {$A_2$};
\node[anchor=north east] at (0,0) {$\zero$};

\draw[thick] (4,3) grid[step=5mm] (8.2,6.2);
\draw[very thick] (4,3) -- (8.2,3);
\draw[very thick] (4,3) -- (4,6.2);
\fill (4,3) circle (2pt);
\node[anchor=north east] at (4,3) {$M_0$};

\node[font=\Large,fill=white] at (6,5) {LCN};

\draw[vector,gray] (2,1) -> (4,2);
\draw[vector,gray] (4,2) -> (6,3);
\draw[vector] (0,0) -> (1,0);
\draw[vector] (1,0) -> (2,1);

\node[fill=white] (start) at (5,1) {Self-starting};
\path[->,>=stealth'] (start.west) edge[bend left] (2.1,0.9);

\draw[vector,gray] (3,4) -> (1.5,2);
\draw[vector,gray] (4.5,6) -> (3,4);
\draw[vector] (1.5,2) -> (1,2);
\draw[vector] (1,2) -> (0,1);
\draw[vector] (0,1) -> (0,0.5);
\draw[vector] (0,0.5) -> (0,0);

\node[fill=white] (stop) at (1.5,4) {Self-stopping};
\path[->,>=stealth'] (stop) edge[bend right] (1.4,2.1);

\end{tikzpicture}

\caption{
Illustration of the reasonning for \pref{thm:irreducibility} on irreducibility.
If the DRN is self-starting, by repeating the reactions, we eventually reach the LCN region from
$\zero$.
In the same manner, if the DRN is self-stopping, we eventually reach $\zero$ from a point in the LCN
region.
If the DRN is LCN irreducible, any point in the LCN region can be reached by any other point in the
LCN region.
In such a setting, one can construct a path from $\zero$ to each elementary vector, and vice-versa.
}
\label{fig:irreducibility}
\end{figure}

\def\sel{\sigma}
\begin{lemma}
\label{lem:self-starting}
$(\exists x \in \Zspos^d \text{ s.t. } \zero\rightarrow^* x) \Longleftrightarrow \exists \sel: [1;d] \mapsto [1;n]$
with:
\begin{enumerate}
\item $\forall k\in[1;d], \exists i\in[1;d], \V_{\sel(i),k} \geq 1$, and
\item  $\O_{\sel(1)} = \zero$ and
		$\forall k\in[2;d], \O_{\sel(k)} \in \mspan_{\Rpos} 
			\left(
\begin{array}{c}
\V_{\sel(1)}\\
\vdots\\
\V_{\sel(k-1)}
\end{array}
\right)$.
\end{enumerate}
\label{lem:0M0}
\end{lemma}
\begin{proof}
($\Leftarrow$) 
Let us define $\forall k\in[1;d], \Omega^k\DEF\{ j\in[1;d]\mid \exists i\in [1;k], \V_{\sel(i)} \geq 1 \}$
and $x^k$ such that $\forall i\in[1;d]$, $x^k_i=1 \EQDEF i\in\Omega^k$ and $x^k_i=0 \EQDEF i\notin\Omega^k$.
We show by induction that $\forall k\in[1;d],\exists x'\succeq x^k\text{ s.t. }\zero\rightarrow^* x'$:
\begin{itemize}

\item $k=1$: $\zero \rightarrow \V_{\sel(1)}$ with $\forall j\in\Omega^1$, $\V_{\sel(1),j} \geq 1$.

\item $k+1$: by induction, (2), and \pref{lem:scale}, $\exists \alpha\in\Zspos$ such that 
$\alpha x^k\geq \O_{\sel(k+1)}$ (with $\zero\rightarrow^* \alpha x^k$). 
Hence, $\alpha x^k\rightarrow \alpha x^k+\V_{\sel(k+1)}$.
We remark that if $\exists i\in \Omega^{k+1}$ such that $(\alpha x^k+\V_{\sel(k+1)})_i < 1$, then necessarily 
$i\in\Omega^k$.
Hence, $\exists \beta\in\Zspos$ such that $(\beta\alpha x^k+\V_{\sel(k+1)})\succeq x^{k+1}$.
Therefore, $\zero \rightarrow^* x'$ with $x'\succeq x^{k+1}$.

\end{itemize}
Finally, as $\Omega^d = [1;d]$, $\exists x\in \Zspos^d\text{ s.t. }\zero\rightarrow^* x$.

($\Rightarrow$)
\def\psel{\varsigma}
$\zero\rightarrow^* x \Rightarrow \exists \ell\in\Zspos,  \exists \pi: [1;\ell] \mapsto [1;n]$ with
	$\sum_{i=1}^\ell \V_{\pi(i)} \in \Zspos^d$,
	and $\forall i\in[1;\ell], \sum_{j=1}^{i-1} \V_{\pi(j)} \succeq \O_{\pi(i)}$.
Let us define $\psel:[1;d]\mapsto[1;\ell]$ iteratively, starting with $\psel(1)\DEF 1$ and $\forall k\in[2;d]$:
\begin{itemize}
\item with $\omega^k\DEF\{ j\in[1;d]\mid \nexists i\in [1;k-1], \V_{\pi(\psel(i)),j} \geq 1 \}$,
\item if $\omega^k = \emptyset$, $\psel(k)\DEF 1$;
\item otherwise, 
$\psel(k) \DEF \min\{ m\in[\psel(k-1)+1;\ell] \mid \exists j\in \omega^k, \V_{\pi(m),j} \geq 1 \}$.
We remark that this minimum necessarily exists (otherwise $x \notin \Zspos^d$), and
$\forall m\in[\psel(k-1);\psel(k)-1]$,
	$\sum_{j=1}^m \pi(j) \in \mspan_{\Rpos}
			\left(
\begin{array}{c}
\V_{\sel(1)}\\
\vdots\\
\V_{\sel(k-1)}
\end{array}
\right)$.
\end{itemize}
From construction, $\sel\DEF\psel\circ\pi$ verifies (1) and (2).

\end{proof}

\begin{theorem}
\label{thm:irreducibility}
DRN $(\V,\O)$ is irreducible if and only if $(\V,\O)$ is LCN irreducible and
	$\exists x\in\Zspos^d \text{ s.t. } \zero \rightarrow^*_{(\V,\O)} x$
	and $\exists x'\in\Zspos^d \text{ s.t. } \zero \rightarrow^*_{(\V,\O)^{-1}} x'$
(i.e. $(\V,\O)$ is self-starting and self-stopping).
\end{theorem}
\begin{proof}
($\Rightarrow$) obvious.

($\Leftarrow$)
If $(\V,\O)$ is LCN irreducible, there exists a minimum origin $M_0\in\Zspos^d$ such that 
$\forall x\succeq M_0,x'\succeq M_0$, $x\rightarrow^*_{(\V,\O)} x'$ and $x'\rightarrow^*_{(\V,\O)} x$.
In addition, $(\V,\O)^{-1}$ is LCN irreducible, with a minimum origin $M'_0\in\Zspos^d$.

From \pref{lem:scale}, $\exists \alpha\in\Zpos$ such that
$\alpha x \succeq M_0$, $\alpha x\succeq M'_0$,
$\alpha x' \succeq M_0$, and $\alpha x'\succeq M'_0$, 
with $\zero\rightarrow^*_{(\V,\O)} \alpha x$ and 
	$\zero\rightarrow^*_{(\V,\O)^{-1}} \alpha x'$.
Hence, $\forall i\in[1;d]$, from \pref{lem:expansion},
\begin{itemize}
\item $\zero\rightarrow^*_{(\V,\O)} \alpha x \rightarrow^*_{(\V,\O)} (\alpha x + e_i) \rightarrow^*_{(\V,\O)} (\alpha x' + e_i)
\rightarrow^*_{(\V,\O)} (\zero + e_i)$, and

\item $(\zero+e_i)\rightarrow^*_{(\V,\O)} (\alpha x + e_i) \rightarrow^*_{(\V,\O)} \alpha x\rightarrow^*_{(\V,\O)} \alpha x'
\rightarrow^*_{(\V,\O)} \zero$.

\end{itemize}
\end{proof}

\begin{example}
One can easily show that the two examples in \pref{fig:examples} are self-starting and
self-stopping.
Using LCN irreducibility criteria from the previous subsection, we conclude that example (b) is
irreducible
(recall that example (a) is not LCN irreducible, so it is not irreducible).
\end{example}

\section{Deciding Recurrence}
\label{sec:recurrence}

Recall that DRN $(\V,\O)$ is recurrent if and only if for all pair of points $x,x'\in\Zpos$, $x\rightarrow^* x'$
implies $x'\rightarrow^* x$ (\pref{def:recurrence}).
First, we show that the LCN recurrence is equivalent to the presence of the null vector in the
strictly positive real span of drift vectors.
Then, we discuss sufficient conditions to obtain the recurrence, and reduce the full
recurrence property to a set of reachability properties.

\subsection{LCN Recurrence}

Let us ignore reaction origins and population positivity constraints.
If $\zero\in\mspan_{\Zspos} \V$, it is clear that from any point $x$, one can undo any reaction
application and then go back to $x$:
$\zero\in\mspan_{\Zspos} \V \Rightarrow \exists \lambda\in\Zspos^n$ such that
$\inner{\lambda}{\V} = \zero$.
Hence $\forall i\in[1;d]$, let us define $\lambda'\in\Zpos^n$ with
$\lambda'_i = \lambda_i - 1$ and $\lambda'_k=\lambda_k, \forall k\in[1;d],k\neq i$:
we obtain $\inner{\lambda'}{\V} = -\V_i$.

By following the proof of \pref{lem:RQ}, we remark in \pref{lem:zRQ} that
$\zero\in\mspan_{\Qspos}\V$ (hence $\zero\in\mspan_{\Zspos}\V$) is equivalent to
$\zero\in\mspan_{\Rspos}\V$.
This can be verified with linear programming.

\begin{lemma}
\label{lem:zRQ}
$\zero\in\mspan_{\Qspos} \V \Longleftrightarrow \zero\in\mspan_{\Rspos}\V$.
\end{lemma}
\begin{proof}
($\Rightarrow$) obvious.
($\Leftarrow$) same proof as for \pref{lem:RQ} with $w=\zero$.
\end{proof}

Finally, \pref{thm:lcn-recurrent} establishes that LCN recurrence is equivalent to
$\zero\in\mspan{\Rspos}\V$.
The main difficulty is to prove that there exists a $M_0\in\Zpos$ such that it is possible to
reverse all the reactions connecting any pair of points above $M_0$ by staying in $\Zpos$.
For that, we consider the basis $\B = \{ b_1, \dots, b_k \}$ of the free $\Z$-module generated by
$\V$.
It is worth noticing that, because $\zero\in\mspan_{\Zspos} \V$, $\forall i\in[1;k],
b_i\in\mspan_{\Zpos} \V$.
Let us pick $M_0$ large enough such that there exists a sequence of reactions from $M_0$ that
can be successively applied (i.e., never below their origins) and that goes to all the vertices
of the fundamental region formed by $\B$ that are adjacent to $M_0$.
Then any pair of points above $M_0$ that is connected can be reversibly reached from each other.
\pref{fig:lcn-recurrent} illustrates this reasoning.

\begin{figure}
\centering

\def\tile{-- ++(1,-1) -- ++(-1,-1) -- ++(-1,1) -- ++(1,1)}

\begin{tikzpicture}

\draw[help lines] (-0.2,-0.2) grid[step=5mm] (8.2,6.2);
\clip (-0.5,-0.5) rectangle (8.2,6.2);

\draw[axe] (0,0) -- (8.2,0);
\draw[axe] (0,0) -- (0,6.2);
\node[anchor=north east] at (8.2,0) {$A_1$};
\node[anchor=north east] at (0,6.2) {$A_2$};
\node[anchor=north east] at (0,0) {$\zero$};

\foreach \y in {3,5,7} { \foreach \x in {3,5,7} { \draw[very thick,gray] (\x,\y) \tile; } }
\foreach \y in {4,6,8} { \foreach \x in {2,4,6,8} { \draw[very thick,gray] (\x,\y) \tile; } }

\foreach \y in {1,3,5,7} { \foreach \x in {1,3,5,7} { \fill (\x,\y) circle (2pt); } }
\foreach \y in {0,2,4,6} { \foreach \x in {0,2,4,6,8} { \fill (\x,\y) circle (2pt); } }

\draw[very thick] (3,3) -- (8.2,3);
\draw[very thick] (3,3) -- (3,6.2);
\fill (3,3) circle (2pt);
\node[anchor=east] at (3,3) {$M_0$};

\end{tikzpicture}

\caption{
Black dots are the points of the lattice generated by $\V$.
The lattice fundamental regions (formed by the basis) are delimited by gray lines.
}
\label{fig:lcn-recurrent}
\end{figure}
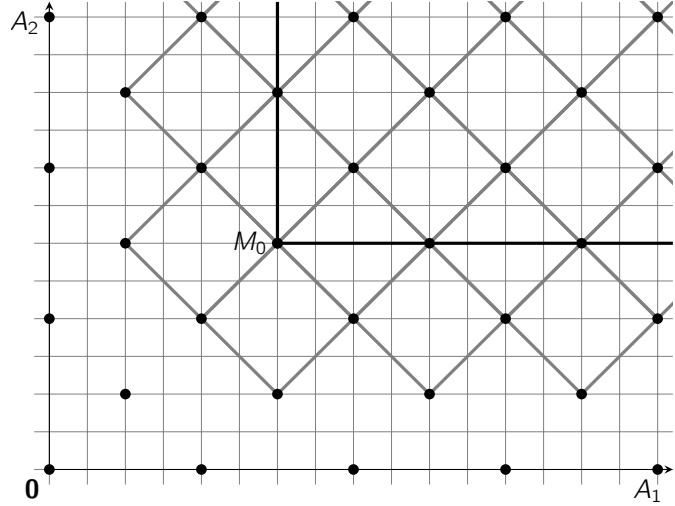

The proof of \pref{thm:lcn-recurrent} also indicates that the reachability graph above $M_0$ is
maximal: if $x+\delta \rightarrow^* x'+\delta$ when $x\succeq M_0, x'\succeq M_0, \delta\in\Zpos^d$,
then $x \rightarrow^* x'$.
This is stated by \pref{cor:saturation}.

\begin{theorem}
\label{thm:lcn-recurrent}
$(\V,\O)$ is LCN recurrent $\Longleftrightarrow
	\zero\in\mspan_{\Rspos} \V$.
\end{theorem}
\begin{proof}
$(\Rightarrow)$ straightforward.

$(\Leftarrow)$
Let us consider $\B = \{b_1, \dots, b_k\}$ the basis of the free $\Z$-module generated by $\V$.

From \pref{lem:zRQ}, $\zero\in\mspan_{\Zspos}\V$, which implies $\forall i\in[1;k], \pm b_i\in\mspan_{\Zpos} \V$.
Hence, $\forall i\in [1;k], \forall s\in \{+,-\}, \exists \lambda^{i,s}\in\Zpos^n$ such that
$\inner{\lambda^{i,s}}{\V} = b^{i,s} \DEF s b_i$.
Let us pick an arbitrary ordering $\pi^{i,s}\in\orderings(\lambda^{i,s})$.

Let us define $M_0\in\Zpos^d$ such that
$\forall \Pi: [1:2k] \mapsto (i,s)$ with
	$i\in[1;k], s\in\{+,-\}$,
	and $\forall l,l'\in[1;2k], \Pi(l) = \Pi(l') \Rightarrow l=l'$,
then
$\forall l\in[1;2k], \forall j\in[1;n],
M_0 +\lp( (\sum_{m=1}^{l-1} b^{\Pi(m-1)}) \apply \pi^{\Pi(m)}) \succeq \O_j$.

From $M_0$ construction, the set of lattice fundamental regions formed by $b1,\dots,b_k$ intersecting $\Z_{\geq M_0}$ is
connected and fits in $\Zpos$.
Moreover, each edge of those fundamental regions can be translated to a sequence of drift vectors $v\in\V$ in
$\Zpos$.
Therefore, $\forall x,x'\succeq M_0$ $x \rightarrow' x' \Rightarrow x' \rightarrow x$.
\end{proof}

\begin{corollary}[Reachability Graph Saturation]
\label{cor:saturation}
If $\zero\in\mspan_{\Zspos} \V$ then there exists $M_0 \in\Zpos^d$ such that the reachability graph
on the set $M_0 + \Zpos^d$ becomes constant in the sense that:
if $x\rightarrow^* x'$, and $x-\delta,x'-\delta\succeq M_0$ for some $\delta\in\Zpos^d$, then
$x-\delta\rightarrow^* x'-\delta$.
\end{corollary}

\begin{example}
From previous section, we know that example (b) of \pref{fig:examples} is irreducible hence
recurrent.
In addition, one can verify that example (a) is LCN recurrent.
\end{example}

\subsection{Full Recurrence}

Assuming DRN $(\V,\O)$ is LCN recurrent, if $\exists x^*\in\Zspos$ such that
$\zero \rightarrow^* x^* \rightarrow^* \zero$, then $(\V,\O)$ is recurrent
(\pref{lem:self-recurrent}).
Indeed, using \pref{lem:scale}, $\exists \alpha\in\Zspos$ such that $\alpha x^* \succeq M_0$.
Then, for any pair of points $x,x'\in\Zpos$, if $x\rightarrow^* x'$, then, by
\pref{lem:expansion}, $x+\alpha x^* \rightarrow^* x'+\alpha x^*$.
Because the DRN is LCN recurrent, $x'+\alpha x^* \rightarrow^* x+\alpha x^*$.
Hence, $x' \rightarrow^* x$.
We remark however that, to our knowledge,  there is no efficient general method to verify if
$\zero \rightarrow^* x^* \rightarrow^* \zero$.

\begin{lemma}
\label{lem:self-recurrent}
If DRN $(\V,\O)$ is LCN recurrent and $\exists x^*\in\Zpos^d$ such that $\zero\rightarrow^* x^*$ and
$x^*\rightarrow^* \zero$, then $(\V,\O)$ is recurrent.
\end{lemma}
\begin{proof}
Let us define $\alpha\in\Zspos$ such that $\alpha x^* \succeq M_0$.
We have the following implication:
$\forall x,x'\in\Zpos^d, x \rightarrow^* x' \Longrightarrow
x' \rightarrow^* x'+\alpha x^* \rightarrow^* x+\alpha x^* \rightarrow^* x\enspace.$
\end{proof}

In the general case, and independently of LCN recurrence, we notice that recurrence is equivalent
to the reachability of the origin of each reaction from the point that is its origin plus drift
vector (\pref{lem:recurrent}).
Again, there is currently no efficient general method to verify these reachability properties.

\begin{lemma}
\label{lem:recurrent}
DRN $(\V,\O)$ is recurrent if and only if $\forall j\in[1;n], \O_j+\V_j \rightarrow^* \O_j$.
\end{lemma}
\begin{proof}
($\Rightarrow$) straightforward.
($\Leftarrow$)
$\forall x\in\Zpos^d, \forall j\in[1;n]: x\succeq \O_j, x\rightarrow x+\V_j\rightarrow^* x$
\end{proof}

The above lemma allows to conclude that any \emph{weakly reversible reaction network} is recurrent
(\pref{lem:wrrn-recurrence}).
A reaction network is weakly reversible if each reaction is part of a cycle of reactions
\cite{JSS-JMC12}; for instance $X\rightarrow Y; Y\rightarrow Z; Z\rightarrow X$ is a weakly
reversible reaction network.

\begin{lemma}
\label{lem:wrrn-recurrence}
Any weakly reversible reaction network is recurrent.
\end{lemma}
\begin{proof}
A DRN models a weakly reversible reaction network if and only if
$\forall j\in[1;n], \exists m\in[1;n]$ and $\pi:[1;m]\mapsto [1;n]$
such that
$\forall k\in[1;m], \O_k = \O_j + \V_j + \sum_{l=1}^{k-1} \V_l$
and
$\O_j = \O_j + \V_j + \sum_{l=1}^{k} \V_l$.
Therefore, $\forall j\in[1;n], \O_j+\V_j \rightarrow^* \O_j$.
\end{proof}

\begin{example}
The sufficient condition for recurrence depicted in \pref{lem:self-recurrent} is verified by 
example (a) of \pref{fig:examples}.
Indeed, $\zero \rightarrow^* (6,6) \rightarrow^* \zero$
(applying $3\V_1$ then $2\V_3$ from $\zero$ results in $(6,6)$, then applying $6\V_2$ results in
$\zero$).
Hence, example (a) is recurrent (but not irreducible), whereas example (b) is irreducible (and recurrent).
\end{example}

\section{Biological Examples}
\label{sec:examples}

This section applies the results of this paper to show that a model of Circadian clock is LCN
irreducible, and a generic model of phosphorylation chain is LCN recurrent.

\subsection{Circadian clock}

We study here a model of PER and TIM circadian oscillations from \cite{Leloup99}, extracted from
the BioModels database \citep{BioModels}.
This model involves 10 species and 26 reactions (including 6 reversible).
The list of reactions is given in \pref{fig:circ}

\begin{figure}
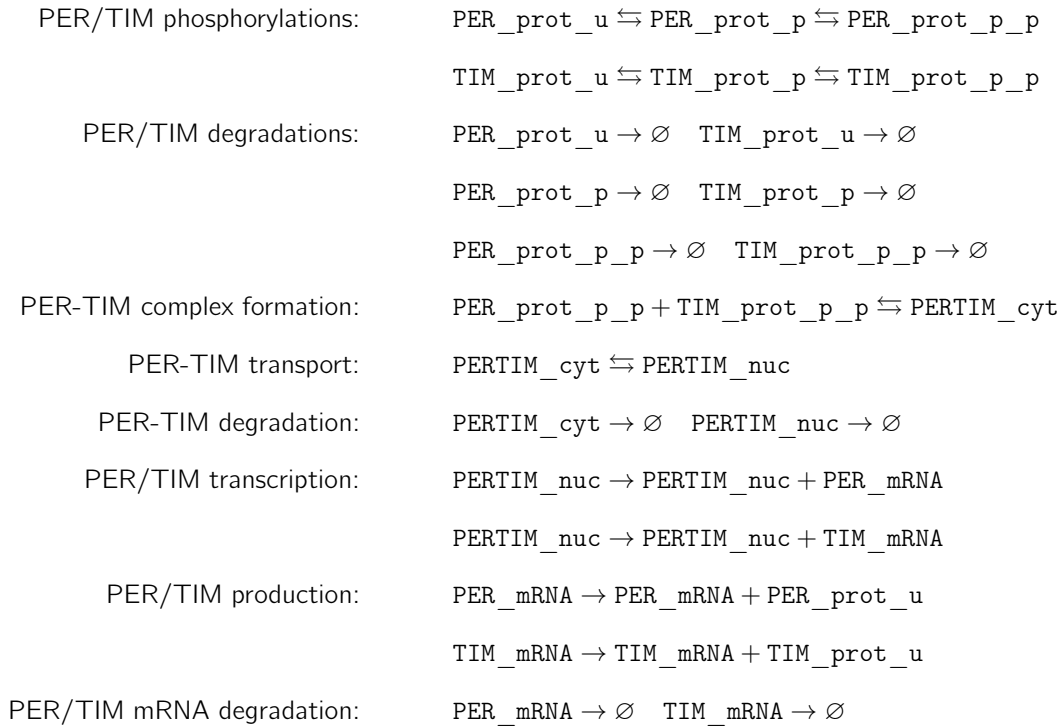

\begin{align*}
\text{PER/TIM phosphorylations:} & &
 & \mathtt{PER\_prot\_u} \leftrightarrows \mathtt{PER\_prot\_p} \leftrightarrows \mathtt{PER\_prot\_p\_p} \\
& & & \mathtt{TIM\_prot\_u} \leftrightarrows \mathtt{TIM\_prot\_p} \leftrightarrows \mathtt{TIM\_prot\_p\_p} \\
\text{PER/TIM degradations:} 
&&& \mathtt{PER\_prot\_u}\rightarrow\varnothing \quad \mathtt{TIM\_prot\_u} \rightarrow \varnothing \\
&&& \mathtt{PER\_prot\_p}\rightarrow\varnothing \quad \mathtt{TIM\_prot\_p} \rightarrow \varnothing \\
&&& \mathtt{PER\_prot\_p\_p}\rightarrow\varnothing \quad \mathtt{TIM\_prot\_p\_p} \rightarrow \varnothing \\
\text{PER-TIM complex formation:}
&&& \mathtt{PER\_prot\_p\_p} + \mathtt{TIM\_prot\_p\_p} \leftrightarrows \mathtt{PERTIM\_cyt}\\
\text{PER-TIM transport:}
&&& \mathtt{PERTIM\_cyt} \leftrightarrows \mathtt{PERTIM\_nuc} \\
\text{PER-TIM degradation:}
&&& \mathtt{PERTIM\_cyt} \rightarrow \varnothing \quad  \mathtt{PERTIM\_nuc}\rightarrow\varnothing \\
\text{PER/TIM transcription:}
&&& \mathtt{PERTIM\_nuc} \rightarrow \mathtt{PERTIM\_nuc} + \mathtt{PER\_mRNA}\\
&&& \mathtt{PERTIM\_nuc} \rightarrow \mathtt{PERTIM\_nuc} + \mathtt{TIM\_mRNA}\\
\text{PER/TIM production:}
&&& \mathtt{PER\_mRNA} \rightarrow \mathtt{PER\_mRNA} + \mathtt{PER\_prot\_u}\\
&&& \mathtt{TIM\_mRNA} \rightarrow \mathtt{TIM\_mRNA} + \mathtt{TIM\_prot\_u}\\
\text{PER/TIM mRNA degradation:}
&&& \mathtt{PER\_mRNA} \rightarrow \varnothing \quad  \mathtt{TIM\_mRNA}\rightarrow\varnothing
\end{align*}
\caption{Reaction network of the PER/TIM circadian oscillations \citep{Leloup99}}
\label{fig:circ}
\end{figure}

One can check that the necessary and sufficient conditions for LCN irreducibility of
\pref{thm:lcn-irreducibility} are verified by this DRN.
Hence, there exists a threshold on the population of species such that there exists a succession of
reactions connecting any pair of states above this threshold.

Because no reaction has an origin being $\zero$, the DRN is not self-starting, hence not fully
irreducible; and because of the presence of degradation reaction, the DRN is not fully
recurrent (for instance, $\zero$ is reachable from the state where all species are 0 except
$\mathtt{PER\_mRNA}$ being 1, but the converse is false).

\subsection{Phosphorylation chains}

We consider a generic model of chains of phosphorylation, where an enzyme $E$ can progressively
phosphorylate a protein up to a certain level $k$.
In concurrence, a kinase $F$ can progressively de-phosphorylate this protein \citep{ALS-MB07}.
\begin{align*}
S_0 + E \leftrightarrows S_0E \rightarrow S_1 + E \leftrightarrows S_1E \rightarrow S_2 + E
\leftrightarrows
\cdots
\rightarrow S_k + E
\\
S_0 + F \leftarrow S_1F \leftrightarrows S_1 + F \leftarrow S_2F \leftrightarrows S_2 + F 
\leftarrow
\cdots
\leftrightarrows S_k + F
\end{align*}

Because of mass conservation properties (notably $\sum_{m=0}^k S_m$ being constant), such a DRN is not
irreducible -- in particular, $\mspan_{\Rspos} \V \neq \R^d$.

Assuming LCN, one can notice that the irreversible reactions such as $S_mE \rightarrow S_{m+1}+E$ can be
undone using the chain of reaction $S_{m+1} + F \rightarrow S_{m+1} F \rightarrow S_m+F$ followed by
$S_m+E\rightarrow S_mE$.
The undo of $S_m + F \leftarrow S_mF$ irreversible reactions is achieved similarly.
This shows that the DRN is LCN recurrent as $\zero\in\mspan_{\Rspos} \V$.
In addition, we remark that it is actually sufficient that all the species are present with at least
one copy in order to undo any irreversible reaction of this network (i.e., $M_0$ can be the vector
having all its components being $1$).

Removing the LCN hypothesis, and in particular considering that $F$ is absent ($0$ copy), it
becomes impossible to revert the reaction $S_0E \rightarrow S_1+E$.
Hence, the DRN is not fully recurrent.

\bigskip

LCN irreducibility depends both on stoichiometry properties (as highlighted by the two 
examples in \pref{fig:examples}) and on the dimension of the lattice generated by $\V$: if the free $\Z$-module
generated by $\V$ has a lower dimension than $\V$, the DRN is not LCN irreducible.
This typically occurs in the presence of mass conservation properties, as highlighted by the
example on phosphorylation chains.

In addition, as stated in \pref{lem:wrrn-recurrence}, we recall that any weakly reversible reaction networks is
recurrent, as the necessarily verify $\zero\in\mspan_{\Rspos}\V$.

\section{Discussion}
\label{sec:discussion}

\paragraph{Relationships between DRNs and stochastic models dynamics}
Markov chains are a widely used modelling framework for analysing dynamics of biochemical reaction networks.
Typically, the discrete states of such Markov chains represent the population of each biochemical species,
and the transitions follow the drift vectors of reactions, when applicable (population of species
greater than the reaction origin).
Then, Markov chains associate either probabilities (DTMCs) or continuous rates (CTMCs) to transitions
following biochemical laws, for instance.

In that sense, a DRN can be considered as the underlying discrete dynamics of \emph{any} Markov
chain modelling the same set of reactions
\citep{FS08}.
If we assume that the probabilities or rates associated to reactions are never null, we obtain the
following correspondence between DRNs and Markov chains dynamical properties:
\begin{itemize}
\item DRN is irreducible if and only if the associated Markov chain is irreducible.
\item DRN is recurrent if and only if all states in the associated Markov chain are recurrent.
\end{itemize}
In the case where probability or rates may become null, DRN irreducibility (resp. recurrence) is still a
necessary condition for Markov chain irreducibility (resp. recurrence).

We note that a DRN which is not recurrent implies that there exists some irreversible steps.
Such a reversible property allows, for instance, an efficient characterization of the stationary
distribution in Markov chains \citep{ACK-BMB10}.

\paragraph{Relationships between DRNs and continuous models dynamics}
Continuous models of reaction networks, such as ODE equations, typically evolve in the continuous
space of concentrations of species and assume that species are present in large copy numbers.
In that way, we may want to relate dynamical properties of such continuous models of reaction
networks to LCN properties of DRNs.

In particular, one can remark that if a DRN is not LCN recurrent, i.e. $\zero\notin\mspan_{\Rspos}
\V$,
there exists a hyperplane in $\R^d$ such that all reaction vectors point in the same side of this
hyperplane, and at least one reaction vector points strictly inside this half-space.
This implies that \emph{no oscillation is possible} in the continuous dynamics: a
non-zero drift is always pushing the system in a constant direction.

\paragraph{Future work}
One possible future direction following the presented results is the derivation of necessary or sufficient
conditions for a discrete definition of \emph{persistence} in continuous models \citep{CNP-JAM12}.
Persistence is the capability for a system to recover a strictly positive population for all species
whenever one the species approaches zero.

One suggested discrete version of this dynamical property is given in \pref{def:persistence}.
We remark that recurrence is a particular case of persistence (\pref{rem:persistence}).

\begin{definition}[Persistence]
\label{def:persistence}
DRN $(\V,\O)$ is \emph{persistent} if and only if
$\forall x\in\Zspos^d, \forall x'\in\Zpos^d$ s.t. $\exists k\in[1;d]. x'_k = 0$,
$x \rightarrow^* x' \Longrightarrow \exists x''\in \Zspos^d. x' \rightarrow x''$
\enspace.
\end{definition}

\begin{remark}
\label{rem:persistence}
Recurrence $\Longrightarrow$ Persistence.
\end{remark}

More generally, the study of Discrete Reaction Networks allows to efficiently prove the absence of certain dynamical properties
in a wide-range of concrete models as they are independent of kinetic parameters.

\section*{Acknowledgements}
The work of GC was supported by NIH grant R01GM086881.

\bibliographystyle{apalike}
\bibliography{biblio}

\end{document}